\pdfoutput=1
\documentclass{amsart}

\usepackage{amsthm}
\usepackage[margin=1cm]{caption}

\usepackage[utf8]{inputenc}
\usepackage[english]{babel}

\usepackage{braket}
\usepackage{amsmath}
\usepackage{amstext}
\usepackage{amssymb}
\usepackage{afterpage}

\usepackage{microtype}

\usepackage{graphicx}

\usepackage{units}
\usepackage{enumerate}
\usepackage{url}
\usepackage{multirow}

\usepackage{tikz}

\usepackage{hyperref}

\allowdisplaybreaks[4]

\newcommand{\declarecommand}[1]{\providecommand{#1}{}\renewcommand{#1}}

\declarecommand{\problem}[4]{%
  \par\vspace{8pt}
  \noindent\refstepcounter{theorem}\textbf{Problem \thetheorem:} #1\textbf{.}\\*[1pt]
  {\begin{tabular}{@{\qquad}lp{0.7\textwidth}}
  \textbf{Parameters:} & #2 \\
  \textbf{Instance:} & #3 \\
  \textbf{Question:} & #4
  \end{tabular}}%
  \par\vspace{8pt}
}

\newtheorem{innercustomthm}{Theorem}
\newenvironment{customthm}[1]
  {\renewcommand\theinnercustomthm{#1}\innercustomthm}
  {\endinnercustomthm}

\newcommand{\Cr}{\text{Cr}}

\newcommand{\M}{\mathcal{M}}

\newcommand{\collapsibility}{\textsc{Col\-lap\-si\-bi\-li\-ty}}
\newcommand{\coll}[2]{$(#1,#2)$-\collapsibility }

\newcommand{\Bigmid}{\;\Big|\;}

\theoremstyle{definition}
\newtheorem{theorem}{Theorem}[section]

\title{Collapsibility to a subcomplex of a given dimension is NP-complete}
\author{Giovanni Paolini}

\begin{document}

\begin{abstract}
  In this paper we extend the works of Tancer, Malgouyres and Francés, showing that \coll{d}{k} is NP-complete for $d\geq k+2$ except $(2,0)$.
By \coll{d}{k} we mean the following problem: determine whether a given $d$-dimensional simplicial complex can be collapsed to some $k$-dimensional subcomplex.
The question of establishing the complexity status of \coll{d}{k} was asked by Tancer, who proved NP-completeness of $(d,0)$ and \coll{d}{1} (for $d\geq 3$).
Our extended result, together with the known polynomial-time algorithms for $(2,0)$ and $d=k+1$, answers the question completely.
\end{abstract}

\maketitle

The final publication is available at Springer via \url{http://dx.doi.org/10.1007/s00454-017-9915-6}.

\section{Introduction} 

Discrete Morse theory is a powerful combinatorial tool which allows to explicitly simplify cell complexes while preserving their homotopy type \cite{batzies2002discrete,chari2000discrete,forman1998morse,kozlov2007combinatorial}.
This is obtained through a sequence of ``elementary collapses'' of pairs of cells.
Such a process might decrease the dimension of the starting complex, or sometimes even leave a single point (in which case we say that the starting complex was collapsible).

The problem of algorithmically recognising collapsibility, or finding ``good'' sequences of elementary collapses, has been studied extensively \cite{benedetti2014random,burton2016parameterized,eugeciouglu1996computationally,joswig2006computing,malgouyres2008determining,tancer2016recognition}.
Such problems proved to be computationally hard even for low-dimensional simplicial complexes.
For 2-dimensional complexes there exists a polynomial-time algorithm to check collapsibility \cite{joswig2006computing,malgouyres2008determining}, but finding the minimum number of ``critical'' triangles (without which the remaining complex would be collapsible) is already NP-hard \cite{eugeciouglu1996computationally}.
In dimension $\geq 3$, collapsibility to some $1$-dimensional subcomplex \cite{malgouyres2008determining} or even to a single point \cite{tancer2016recognition} were proved to be NP-complete.

In \cite{tancer2016recognition}, Tancer also introduced the general \coll{d}{k} problem: determine whether a $d$-dimensional simplicial complex can be collapsed to some $k$-dimensional subcomplex.
He showed that \coll{d}{k} is NP-complete for $k\in\{0,1\}$ and $d\geq 3$, extending the result of Malgouyres and Francés about NP-completeness of \coll{3}{1} \cite{malgouyres2008determining}.
Tancer also pointed out that the codimension $1$ case ($d=k+1$) is polynomial-time solvable as is the $(2,0)$ case.
He left open the question of determining the complexity status of \coll{d}{k} in general.

In this short paper we extend Tancer's work, and prove that \coll{d}{k} is NP-complete in all the remaining cases.

\begin{customthm}{\ref{thm:main-theorem}}
  The \coll{d}{k} problem is NP-complete for $d\geq k+2$, except for the case $(2,0)$.
\end{customthm}

To do so, we prove that \coll{d}{k} admits a polynomial-time reduction to \coll{d+1}{k+1} (Theorem \ref{thm:reduction}).
Then the main result follows by induction.
The base cases of the induction are given by NP-completeness of \coll{3}{1} (for codimension $2$) and of \coll{d}{0} (for codimension $d\geq 3$).

\section{Collapsibility and discrete Morse theory}

We refer to \cite{hatcher} for the definition and the basic properties of simplicial complexes, and to \cite{kozlov2007combinatorial} for the definition of elementary collapses.
The simplicial complexes we consider do not contain the empty simplex, unless otherwise stated.

Our focus is the following decision problem.

\problem{\coll{d}{k}}
{Non-negative integers $d > k$.}
{A finite $d$-dimensional simplicial complex $X$.}
{Can $X$ be collapsed to some $k$-dimensional subcomplex?}

We are now going to recall a few definitions of discrete Morse theory \cite{chari2000discrete,forman1998morse,kozlov2007combinatorial}, so that we can restate the \coll{d}{k} problem in terms of acyclic matchings.

Given a simplicial complex $X$, its \emph{Hasse diagram} $H(X)$ is a directed graph in which the set of nodes is the set of simplexes of $X$, and an arc goes from $\sigma$ to $\tau$ if and only if $\tau$ is a face of $\sigma$ and $\dim(\sigma) = \dim(\tau) + 1$.
We denote such an arc by $\sigma \to \tau$.
A \emph{matching} $\M$ on $X$ is a set of arcs of $H(X)$ such that every node of $H(X)$ (i.e.\ every simplex of $X$) is contained in at most one arc in $\M$.
Given a matching $\M$ on $X$, we say that a simplex $\sigma\in X$ is \emph{critical} if it does not belong to any arc in $\M$.
Finally we say that a matching $\M$ on $X$ is \emph{acyclic} if the graph $H(X)^\M$, obtained from $H(X)$ by reversing the direction of each arc in $\M$, does not contain directed cycles.

Notice that the empty set is always a valid acyclic matching, for which all simplices are critical.
See Fig.\ \ref{fig:matching} for an example of a non-trivial acyclic matching on the full triangle.

\begin{figure}[htbp]
  \begin{center}
  \begin{tikzpicture}[->]
    \node (123) at (0,0) {$\{1,2,3\}$};
    \node (12) at (-1,-1) {$\{1,2\}$};
    \node (13) at (0,-1) {$\{1,3\}$};
    \node (23) at (1,-1) {$\{2,3\}$};
    \node (1) at (-1,-2) {$\{1\}$};
    \node (2) at (0,-2) {$\{2\}$};
    \node (3) at (1,-2) {$\{3\}$};
    
    \draw (123) -- (12);
    \draw (123) -- (13);
    \draw[line width=0.5mm] (123) -- (23);
    \draw (12) -- (1);
    \draw[line width=0.5mm] (12) -- (2);
    \draw (13) -- (1);
    \draw (13) -- (3);
    \draw (23) -- (2);
    \draw (23) -- (3);
  \end{tikzpicture}
  \end{center}
  \caption{An acyclic matching on the full simplicial complex on $3$ vertices, with critical simplices $\{1,3\}$, $\{1\}$, $\{3\}$.}
  \label{fig:matching}
\end{figure}
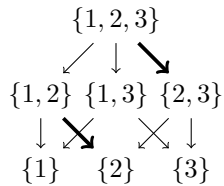

By standard facts of discrete Morse theory (see for instance \cite[Section 11.2]{kozlov2007combinatorial}), ``collapsibility to some $k$-dimensional subcomplex'' is equivalent to ``existence of an acyclic matching such that the critical cells form a $k$-dimensional subcomplex''.
Notice that, given an acyclic matching $\M$ with no critical simplices of dimension $> k$, one can always remove from $\M$ the arcs between simplices of dimension $\leq k$ and obtain an acyclic matching where the critical simplices form a $k$-dimensional subcomplex.
Therefore the collapsibility problem can be restated as follows.

\problem{\coll{d}{k} (equivalent form)}
{Non-negative integers $d > k$.}
{A finite $d$-dimensional simplicial complex $X$.}
{Does $X$ admit an acyclic matching such that all critical simplices have dimension $\leq k$?}


To simplify the proof of Theorem \ref{thm:reduction} we quote the following useful result from \cite{kozlov2007combinatorial}, adapting it to our notation.

\begin{theorem}[Patchwork theorem {\cite[Theorem 11.10]{kozlov2007combinatorial}}]
  Let $P$ be a poset.
  Let $\varphi\colon X \to P$ be an order-preserving map (where $X$ is ordered by inclusion), and assume to have acyclic matchings on subposets $\varphi^{-1}(p)$ for all $p\in P$.
  Then the union of these matchings is itself an acyclic matching on $X$.
  \label{thm:patchwork-theorem}
\end{theorem}

Notice that the subposets $\varphi^{-1}(p)$ are not subcomplexes of $X$ in general, but they still have a well-defined Hasse diagram (the induced subgraph of $H(X)$).
Thus all the previous definitions (matching, critical simplex, acyclic matching) apply also to each subposet.

\section{Main result}

\begin{theorem}
  Let $d>k\geq 0$.
  Then there is a polynomial-time reduction from \coll{d}{k} to \coll{d+1}{k+1}.
  \label{thm:reduction}
\end{theorem}

\begin{proof}
  Let $X$ be an instance of \coll{d}{k}, i.e.\ a $d$-dimensional simplicial complex. Let $V = \{v_1, \dots, v_r\}$ be the vertex set of $X$.
  Construct an instance $X'$ of \coll{d+1}{k+1}, i.e.\ a $(d+1)$-dimensional complex, as follows.
  Let $n\geq 1$ be the number of simplices in $X$.
  Introduce new vertices $w_1, \dots, w_{n+1}$, and define $X'$ as the simplicial complex on the vertex set $V' = \{v_1, \dots, v_r, w_1, \dots, w_{n+1}\}$ given by
  \[ X' = X \cup \Big\{ \sigma \cup \{w_i\} \Bigmid \sigma \in X, \; i=1,\,\dots,\,n+1 \Big\}. \]
  Then $X'$ has $n(n+2)$ simplices.
  Roughly speaking, $X'$ is obtained from $X$ by attaching $n+1$ cones of $X$ to $X$.
  We are going to prove that $X$ is a yes-instance of \coll{d}{k} if and only if $X'$ is a yes-instance of \coll{d+1}{k+1}.
  
  Suppose that $X$ is a yes-instance of \coll{d}{k}.
  Then there exists an acyclic matching $\M$ on $X$ such that all critical simplices have dimension $\leq k$.
  Construct a matching $\M'$ on $X'$ as follows:
  \begin{eqnarray*}
    \M' &=& \Big\{ \sigma \cup \{ w_1 \} \to \sigma \Bigmid \sigma \in X \Big\} \cup \\
           && \Big\{ \sigma \cup \{ w_i \} \to \tau \cup \{ w_i \} \Bigmid (\sigma \to \tau) \in \M, \; i=2,\,\dots,\,n+1 \Big\}.
  \end{eqnarray*}
  This matching corresponds to collapsing the first cone together with $X$, and every other ``base-less'' cone by itself (as a copy of $X$).
  Notice that the critical simplices of $\M'$ do not form a subcomplex of $X'$, even when the critical simplices of $\M$ form a subcomplex of $X$.  
  
  To prove that $\M'$ is acyclic, consider the set $P = \{ w_1, \dots, w_{n+1} \}$ with the partial order
  \[ w_i < w_j \, \text{ if and only if } \, i=1 \text{ and } j>1. \]
  Let $\varphi\colon X' \to P$ be the order-preserving map given by
  \[ \varphi(\sigma) =
     \begin{cases}
       w_j & \text{if $\sigma$ contains $w_j$ for some $j\geq 2$;}\\
       w_1 & \text{otherwise.}
     \end{cases}
  \]
  Then $\M'$ is a union of matchings $\M_j'$ on each fiber $\varphi^{-1}(w_j)$. The matching $\M_1'$ is acyclic on $\varphi^{-1}(w_1)$, since the arcs of $\M_1'$ define a cut of the Hasse diagram of $\varphi^{-1}(w_1)$.
  The Hasse diagram of each $\varphi^{-1}(w_j)$ for $j\geq 2$ is isomorphic to $H(X \cup \{\varnothing\})$ via the map $\sigma \cup \{w_j\} \mapsto \sigma$, and the matching $\M_j$ maps to $\M$. Since $\M$ is acyclic on $H(X)$, each $\M_j$ is also acyclic on $\varphi^{-1}(w_j)$.
  By the Patchwork theorem (Theorem \ref{thm:patchwork-theorem}), $\M'$ is acyclic on $X'$.

  The set of critical simplices of $\M'$ is
  \[ \Cr(X', \M') = \{ w_1 \} \cup \Big\{ \sigma \cup \{ w_i \} \Bigmid \sigma \in \Cr(X, \M) \cup \{\varnothing\}, \; i=2,\,\dots,\,n+1 \Big\}. \]
  In particular, all critical simplices have dimension $\leq k+1$.
  Therefore $X'$ is a yes-instance of \coll{d+1}{k+1}.
  
  Conversely, suppose now that $X'$ is a yes-instance of \coll{d+1}{k+1}.
  Let $\M'$ be an acyclic matching on $X'$ such that all critical simplices have dimension $\leq k+1$.
  Since $X$ contains $n$ simplices, and there are $n+1$ cones, there must exist an index $j\in\{1,\,\dots,\,n+1\}$ such that
  \[ \Big(\sigma \cup \{w_j\} \to \sigma \Big) \not\in \M' \quad \forall\; \sigma \in X. \]
  In other words, simplices containing $w_j$ are only matched with simplices containing $w_j$.
  Then we can construct a matching $\M$ on $X$ as follows:
  \[ \M = \Big\{ \sigma \to \tau \Bigmid \sigma, \tau \in X \text{ satisfying } \Big(\sigma \cup \{ w_j \} \to \tau \cup \{ w_j \}\Big) \in \M' \Big\}. \]
  Notice that if there is some $0$-dimensional simplex $\sigma = \{v\}\in X$ such that $(\{v, w_j\} \to \{w_j\}) \in \M'$, then $\{v\}$ is critical with respect to $\M$ (it would be matched with $\tau=\varnothing$, which does not exist in $X$).
  The Hasse diagram of $X$ injects into the Hasse diagram of the $j$-th cone via the map
  \[ \iota\colon \sigma \mapsto \sigma \cup \{ w_j \}, \]
  and by construction $\M$ maps to $\M'$.
  Since $\M'$ is acyclic, $\M$ is also acyclic.
  The set of critical simplices of $\M$ is
  \[ \Cr(X, \M) = \Big\{ \sigma \in X \Bigmid \sigma \cup \{ w_j \} \in \Cr(X', \M') \,\text{ or }\, \Big(\sigma\cup\{w_j\} \to \{w_j\}\Big) \in \M' \Big\}. \]
  In the first case $\sigma \cup \{w_j\}$ has dimension $\leq k+1$, and in the second case $\sigma$ is $0$-dimensional.
  In particular, all critical simplices have dimension $\leq k$.
  Therefore $X$ is a yes-instance of \coll{d}{k}.
\end{proof}

The \coll{d}{k} problem admits a polynomial-time solution when $d=k+1$ and also for the case $(2,0)$ \cite{joswig2006computing,malgouyres2008determining,tancer2016recognition}.
Malgouyres and Francés \cite{malgouyres2008determining} proved that \coll{3}{1} is NP-complete, and Tancer \cite{tancer2016recognition} extended this result to \coll{d}{k} for $k\in \{0,\, 1\}$ and $d\geq 3$.
Using this as the base step and Theorem \ref{thm:reduction} as the induction step, we obtain the following result.

\begin{theorem}
  The \coll{d}{k} problem is NP-complete for $d\geq k+2$, except for the case $(2,0)$. \qed
  \label{thm:main-theorem}
\end{theorem}

\section{Acknowledgements}
I would like to thank my father, Maurizio Paolini, for giving useful comments and suggesting corrections.
I would also like to thank Luca Ghidelli, for checking the proof thoroughly and for being my best man.
Finally, I would like to thank the anonymous referee for his/her suggestions.

\bibliographystyle{amsalpha-abbr}
\bibliography{bibliography}

\end{document}